\documentclass[11pt]{article}
\usepackage{indentfirst}
\usepackage{epsfig}
\usepackage{amssymb,amsmath,amsfonts,mathrsfs}
\usepackage{subfigure}
\usepackage{cite}
\usepackage{float}
\usepackage{extarrows}
\usepackage{multirow}
\usepackage{color}

\newenvironment{mysummary}[1]{%
    \leftskip=2.5em \rightskip=2.5em
    \noindent\small{\bfseries #1}}
    {\par\medskip}
\renewenvironment{abstract}{\begin{mysummary}{Abstract:}}{\end{mysummary}}
\newenvironment{keywords}{\begin{mysummary}{Key words:}}{\end{mysummary}\medskip}
\newenvironment{proof}{\par{\itshape Proof}.\ }
    {\hfill\raisebox{.56ex}{\fbox{}}\par}
\newtheorem{theorem}{\indent Theorem}
\newtheorem{lemma}{\indent Lemma}

\newtheorem{definition}{\indent Definition}

\newtheorem{myremark}{\indent Remark}
\newenvironment{remark}{\begin{myremark}\normalfont}
    {\end{myremark}}
\newtheorem{myexample}{\indent Example}

\title{\textbf{Observer-Based Drag-Tracking Guidance for Entry Vehicles Considering Input Saturation Constraint}\thanks{This work is supported by National Natural Science Foundation of China (61873029, 61873250) and Beijing Natural Science Foundation (4192068).
}
}

\author{Han Yan\thanks{Senior engineer, E-mail addresses: yhustc@sina.com.}, Yingzi He\thanks{Professor.}, Chunling Wei\thanks{Professor.}\\[1ex]Science and Technology on Space
Intelligent Control Laboratory,\\
Beijing Institute of Control Engineering,
Beijing 100190, China}
\date{}
\begin{document}

\maketitle

\begin{abstract}
This paper studies the drag-tracking guidance design problem of uncertain entry vehicles. With employing a Nussbaum type function to deal with input saturation constraint, an output feedback guidance law (bank angle magnitude)  with a high-gain observer is constructed  that makes the drag-tracking error converge near zero in the presence of uncertainties. It is also worthy to claim that, in contrast to the existing results whose envelope of uncertainty merely depends on the drag error, the considered uncertainty is allowed to be not bigger than a function of drag error and integral term of drag error, which inevitably occurs in practice. The Monte Carlo simulation is done to illustrate  the advantage of the developed method.
\end{abstract}


\section{Introduction}
Due to the advantages over predictor-corrector guidance in realization, reference-trajectory guidance has been extensively investigated and applied in practice. Lots of reference-trajectory tracking guidance synthesis strategies have been developed based on various methods, such as indirect Legendre pseudospectral method \cite{Tian_2011}, trajectory linearization control (TLC) approach \cite{ISA_2014,ISA_2015,AESCTE_2015}, and small-gain theorem \cite{AA_2017}. A  typical  reference-trajectory guidance is to make the vehicle track the drag profile that generated from the reference-trajectory, and  has also been validated in the Apollo and Shuttle Programs \cite{Shuttle}. The successful application further promotes the research on it. In order to improve robustness and guidance precision, many modern control methods have been used to design drag-tracking guidance law, including feedback linearization method \cite{Leavitt_2007,Mease_2002,Wang_2013,Saraf_2004,Talole_2007}, predictive control \cite{Guo,Lu_1997,Benito_2008}, and active disturbance rejection control (ADRC) \cite{XiaYuanQing}.

Most of drag-tracking guidance laws require the knowledge of drag rate, which is hard for a vehicle to measure accurately in practice. To deal with above challenge, the altitude rate is used as feedback instead of the drag rate in the Shuttle guidance, but it is error prone as mentioned in \cite{Shuttle}. Recently, several observer-based techniques have been  used to estimate the drag rate. The sliding mode state and perturbation observer is used in \cite{Talole_2007} to address the issue of estimation of the drag rate. In \cite{Xia_2013}, an extended state observer is introduced to estimate the drag rate and an extended state, and the ADRC algorithm is utilized to design a drag-tracking law. It is worthy to mention that the results in \cite{Talole_2007} and \cite{Xia_2013} are based on the assumption that the uncertainty is bounded. However, since the uncertainty term relies on drag and other states of vehicle, the boundedness cannot be guaranteed. In comparison, after analyzing the uncertainties, \cite{Yan} assumes that the uncertainty term is not bigger than a linear function of absolute value of drag-tracking error, and the input-to-state stability theory is applied to design a state feedback guidance law. Moreover, the high-gain observer is also used in \cite{Yan} to estimate the drag rate, recovering the performance of the state feedback law.  But strictly speaking, the assumption in \cite{Yan} maybe still not hold in practice because the uncertainty term is relative to velocity of vehicle.

In the entry process, the bank angle is the only control variable that can be modulated to eliminate the drag-tracking error, and since the magnitude of bank angle is limited, the guidance law should be designed in the presence of input saturation. However, to the best knowledge of the authors, there are few attempts are made on guidance law design for entry vehicles with magnitude constraints. In \cite{Wen_2011}, a robust control design method is proposed for single input uncertain nonlinear systems in the presence of input saturation. The Nussbaum function is introduced to compensate for the nonlinear term arising from the saturation constraints, and the stability analysis is given in the framework of backstepping scheme. Such an idea is also used in the integrated guidance and control design with taking the saturation of the actuators into account \cite{Liang_2015}.

In this paper, a general uncertainty term and input saturation is considered in modelling the  drag dynamics. As the drag rate maybe difficult for a vehicle to measure accurately, a high-gain observer is used in the designing of the guidance law without drag rate measurement. A Nussbaum type function is also introduced to deal with the input saturation, and the stability analysis shows that the guidance law can make the drag-tracking error converge sufficiently near zero. The contribution of the paper is summarized  as follows. First, comparing with our former work (i.e. \cite{Yan}), an auxiliary   integral term of drag error  is introduced in the guidance law, which also leads to the fact that the assumption on uncertainty term can be further weakened. Second, an output feedback guidance law with high-gain observer \cite{Khalil_1992,Khalil_1999} is designed  such that the drag-tracking error can converge to a small residual set around the origin. Third, inspiring from the idea in \cite{Wen_2011,Liang_2015}, the input saturation problem is considered in the process of designing the guidance law.

The remainder of this paper is organized as follows.  Section \ref{Preliminary} presents some preliminaries. The drag dynamics is formulated in Section \ref{sec of model}. Section \ref{guidance law design} elaborates the output feedback guidance law in the presence of input saturation for entry. Section \ref{simulation} shows the simulation results. Finally, Section \ref{conclusion} summarizes the conclusions.
\section{Preliminary}\label{Preliminary}
\begin{definition}
Any continuous function $N(s): R\rightarrow R$ is a function of Nussbaum type if it has the following properties:
\begin{equation}\label{Nussbaum property1}
\lim_{s\rightarrow\infty} \sup\frac{1}{s}\int_{0}^{s}N(\zeta)\mathrm{d}\zeta=+\infty
\end{equation}
\begin{equation}\label{Nussbaum property2}
\lim_{s\rightarrow\infty} \inf\frac{1}{s}\int_{0}^{s}N(\zeta)\mathrm{d}\zeta=-\infty
\end{equation}
\end{definition}

\begin{lemma}\label{Nussbaum lemma}
Let $V(\cdot)$ and $\mathcal{X}(\cdot)$ be smooth functions defined on $[0,t_f]$ with $V(t)\geq0$, $\forall t\in[0,t_f]$, and $N(\mathcal{X})=e^{\mathcal{X}^2}\cos\left(\frac{\pi}{2}\mathcal{X}\right)$ is an even smooth Nussbaum type function. The following inequality holds:
\begin{equation}\label{condition for lemma1}
0\leq V(t)\leq c_0 + c_1\int_0^t\left(\rho(\tau)N(\mathcal{X}(\tau))\dot{\mathcal{X}}(\tau)-\dot{\mathcal{X}}(\tau)\right)e^{\lambda (\tau-t)}\mathrm{d}\tau
\end{equation}
where $\lambda$, $c_0$ and $c_1$ are constants, $c_1>0,\lambda>0$, and $\rho(t)$ is a time-varying parameter that satisfies $0<\rho_{\min}\leq \rho(t)\leq \rho_{\max}$. Then, $\mathcal{X}(t)$, $V(t)$ and $\int_0^t\left(\rho(\tau)N(\mathcal{X})\dot{\mathcal{X}}-\dot{\mathcal{X}}\right)e^{\lambda (\tau-t)}\mathrm{d}\tau$ must be bounded on $[0,t_f]$.
\end{lemma}
\begin{proof}
See the Appendix.
\end{proof}
\section{Model Derivation}\label{sec of model}
The motion equations of an unpowered, point mass vehicle flying over a non-rotating planet in a stationary atmosphere are given by \cite{Talole_2007,Xia_2013,Tian_2011,Guo}
\begin{subequations}\label{model}
\begin{equation}\label{dot_r}
\dot{r}=v\sin\gamma
\end{equation}
\begin{equation}
\dot{\phi}=\frac{v\cos\gamma\sin\chi}{r\cos\theta}
\end{equation}
\begin{equation}
\dot{\theta}=\frac{v\cos\gamma\cos\chi}{r}
\end{equation}
\begin{equation}\label{dot_v}
\dot{v}=-D-g\sin\gamma
\end{equation}
\begin{alignat}{1}
\dot{\gamma}=&\frac{L\cos\sigma}{v}-\left(\frac{g}{v}-\frac{v}{r}\right)\cos\gamma
\end{alignat}
\begin{equation}
\dot{\chi}=\frac{L\sin\sigma}{v\cos\gamma}+\frac{v\cos\gamma\sin\chi\tan\theta}{r}
\end{equation}
\end{subequations}
where $r$ is the radial position, $\phi$ is longitude, $\theta$ latitude, $v$ is the velocity, $\gamma$ is the flight path angle, $\chi$ is the heading angle, $L$ is the lift acceleration, $D$ is the drag acceleration, and $g$ is gravitational acceleration. The downrange $s$ can be calculated according to the states of (\ref{model}), and one also has the differential equation for $s$ as \cite{Lu}
\begin{equation}
\dot{s}=-\frac{vr_0\cos\gamma}{r}
\end{equation}
where $r_0$ is the reference radius. $L$ and $D$ can be calculated as
\begin{subequations}
\begin{equation}
L=\frac{1}{2m}\rho v^{2}S(\underbrace{C_{L}^{0}+\Delta C_{L}}_{C_{L}})
\end{equation}
\begin{equation}\label{eq of drag}
D=\frac{1}{2m}\rho v^{2}S(\underbrace{C_{D}^{0}+\Delta C_{D}}_{C_{D}})
\end{equation}
\end{subequations}
where $m$ is the vehicle mass, $\rho$ is the atmospheric density, $S$ is the reference area, $C_{L}^{0}$ and $C_{D}^{0}$ are nominal values of aerodynamic coefficients, and $\Delta C_{L}$ and $\Delta C_{D}$ are bounded uncertainties. An exponential atmospheric density model
\begin{equation}\label{exponential atmospheric density model}
\rho=\rho_{0}e^{-\frac{h}{h_{s}}}+\Delta \rho
\end{equation}
is assumed, where $h=r-r_{0}$, $r_{0}$ is the reference radius, $\rho_{0}$ is atmospheric density at the reference radius, $\Delta \rho$ is bounded uncertainty, and $h_{s}$ is characteristic constant. The gravitational acceleration as a function of $r$ is given by
\begin{equation}
g=\frac{\mu}{r^{2}}
\end{equation}
where $\mu$ is gravitational constant.

Due (\ref{eq of drag}), one has
\begin{equation}\label{dot_D}
\dot{D}=\frac{1}{2}\dot{\rho}v^{2}C_{D}\frac{S}{m}+\rho v\dot{v}C_{D}\frac{S}{m}+\frac{1}{2}\rho v^{2}\dot{C}_{D}\frac{S}{m}
\end{equation}
and
\begin{equation}
\frac{\dot{D}}{D}=\frac{\dot{\rho}}{\rho}+\frac{2\dot{v}}{v}+\frac{\dot{C}_{D}}{C_{D}}
\end{equation}
It can also be calculated out that
\begin{subequations}
\begin{equation}\label{dot_rho}
\frac{\dot{\rho}}{\rho}=-\frac{\dot{h}}{h_{s}}+\delta_{\rho}=-\frac{\dot{r}}{h_{s}}+\delta_{\rho}\xlongequal{\mathrm{Eq.}~(\ref{dot_r})}-\frac{v\sin\gamma}{h_{s}}+\delta_{\rho}
\end{equation}
\begin{equation}\label{dot_C_D}
\frac{\dot{C}_{D}}{C_D}=\frac{\dot{C}_{D}^{0}}{C_D^{0}}+\delta_{C_D}
\end{equation}
\begin{equation}
\dot{g}=-\frac{2\mu}{r^{3}}v\sin\gamma=-\frac{2gv\sin\gamma}{r}
\end{equation}
\end{subequations}
where $\delta_{\rho}=\frac{\rho\Delta\dot{\rho}-\dot{\rho}\Delta\rho}{\rho(\rho-\Delta\rho)}$ and $\delta_{C_D}=\frac{\Delta\dot{C}_{D}C_{D}^{0}-\Delta C_{D}\dot{C}_{D}^{0}}{C_{D}^{0}(C_{D}^{0}+\Delta C_{D})}$. Thus,
\begin{equation}\label{dotD}
\frac{\dot{D}}{D}=-\frac{v\sin\gamma}{h_s}-\frac{2D}{v}-\frac{2g\sin\gamma}{v}+\underbrace{\frac{\dot{C}_{D}^{0}}{C_{D}^{0}}}_{C}+\underbrace{\delta_{\rho}+\delta_{C_D}}_{\delta}
\end{equation}
Furthermore,
\begin{equation}\label{dotD}
\dot{D}=p(D,t)+\delta D
\end{equation}
\begin{equation}\label{ddotD}
\ddot{D}=f(D,t)+g_{0}(D,t)\mathrm{sat}(u)+\Delta(D,t)
\end{equation}
where
\begin{equation}
\mathrm{sat}(u)=\cos\sigma
\end{equation}
and
\begin{equation}\nonumber
p=\left(-\frac{v\sin\gamma}{h_s}-\frac{2D}{v}-\frac{2g\sin\gamma}{v}+C\right)D
\end{equation}
{\footnotesize\begin{alignat}{1}
f=&\left(-\frac{v\sin\gamma}{h_{s}}-\frac{4D}{v}-\frac{2g\sin\gamma}{v}+C\right)\left(-\frac{v\sin\gamma}{h_s}D-\frac{2D^2}{v}-\frac{2g\sin\gamma}{v}D+CD\right)\nonumber\\
&+D\left(\frac{D\sin\gamma+g}{h_{s}}+\frac{4g\sin^{2}\gamma-2g\cos^{2}\gamma}{r}-\frac{2D^{2}+4Dg\sin\gamma+2g^{2}\sin^{2}\gamma-2g^{2}\cos^{2}\gamma}{v^2}+\frac{v^{2}\cos^{2}\gamma}{rh_{s}}+\dot{C}\right)\nonumber
\end{alignat}}
\begin{equation}\nonumber
g_{0}=-\left(\frac{v}{h_{s}}+\frac{2g}{v}\right)\frac{LD\cos\gamma}{v}
\end{equation}
\begin{equation}\nonumber
\Delta=D(\dot{\delta}+\delta^2)+\underbrace{\left(-\frac{2v\sin\gamma}{h_{s}}-\frac{6D^{2}}{v}-\frac{4g\sin\gamma}{v}D+2CD\right)}_{Q(D,t)}\delta
\end{equation}
Since the purpose of designing a guidance law is to make the drag acceleration $D$ track its reference value $D^{*}$ by modulating the bank angle $\sigma$, we define $\widetilde{D}=D-D^{*}$ and $x=[x_{0},x_{1},x_{2}]^{T}=\left[\int\widetilde{D}\mathrm{d}t,\widetilde{D},\dot{\widetilde{D}}\right]^{T}$. The drag dynamics for guidance law design is formulated as
\begin{equation}\label{sample ISS system}
\dot{x} =Ax+B(f(D,t)-\ddot{D}^{*}+g_{0}(D,t)\mathrm{sat}(u)+\Delta(D,t))
\end{equation}
where
\begin{equation}\nonumber
A=\left[\begin{array}{c|c}
0 & I_{2}\\ \hline
0 &0\end{array}\right],~B=\begin{bmatrix}0&0&1
\end{bmatrix}^{T}
\end{equation}
To facilitate control system design, the saturation nonlinearity is approximated by a smooth function defined as
\begin{equation}\label{input with tanh}
g(u)=\tanh(u)=\frac{e^u-e^{-u}}{e^u+e^{-u}}
\end{equation}
Then one has
\begin{equation}
\mathrm{sat}(u)=g(u)+\bar{d}(t)
\end{equation}
where it can be easily to verify that $\bar{d}(t)$ is a bounded function. In addition, the authors introduce the following auxiliary system
\begin{equation}\label{auxiliary system}
\dot{u}=-\frac{1}{\tau}u+\frac{1}{\tau}u_{c}
\end{equation}
where $\tau$ is a positive filter time constant, $u_c$ is an auxiliary control signal. Thus, the overall system consisting of auxiliary system (\ref{auxiliary system}) and system (\ref{sample ISS system}) can be rewritten as
\begin{subequations}\label{sample ISS system with sat}
\begin{equation}\label{sample ISS system for state1}
\dot{x} =Ax+B(f(D,t)-\ddot{D}^{*}+g_{0}(D,t)g(u)+\Delta_s(D,t))
\end{equation}
\begin{equation}\label{sample ISS system for state2}
\dot{u}=-\frac{1}{\tau}u+\frac{1}{\tau}u_{c}
\end{equation}
\end{subequations}
where $\Delta_s(D,t)=D(\dot{\delta}+\delta^2)+Q(D,t)\delta+g_0(D,t)\bar{d}(t)$. Here, we assume that uncertainties $\delta$ and $\dot{\delta}$ are bounded. It can be seen that the expression of $\Delta_s(D,t)$ contains $v$ and $D$, and Eq. (\ref{dot_v}) implies that $v$ is a function of $\int D \mathrm{d}t$. Therefore, in a reasonable flight domain of interest there exist positive constants $l_0$, $l_1$, $l$ and $d$ such that
\begin{equation}\label{assumption}
|\Delta_s|\leq l_0|x_{0}|+l_1|x_{1}|+d\leq l\|x\|+d
\end{equation}
holds. Beside, $-90^{\circ}<\gamma<90^{\circ}$ is also assumed in the domain. From this, clearly, $g_0$ is invertible.
\begin{remark}
A drag-tracking guidance law has been got in \cite{Yan} under the assumption that the uncertainties related term $\Delta$ is not bigger than a linear function of $|\widetilde{D}|$. However, since $\Delta$ is also a function of $v$ as we mentioned, strictly speaking, that assumption of $\Delta$ cannot be guaranteed. Comparing with \cite{Yan}, the assumption as expressed by Eq. (\ref{assumption}) is more reasonable.
\end{remark}
\begin{remark}
In our case, the control signal to be designed is $u$. (\ref{sample ISS system for state2}) is artificially introduced to generate a stable control signal $u$ by designing an auxiliary control signal $u_c$.
\end{remark}
\begin{remark}
We can see that Eq. (\ref{sample ISS system with sat}) is not in the standard form of chains of integrators, as $\dot{x}$ is related to nonlinear function $g(u)$ instead of $u$ directly. This results in a term $\frac{\partial g(u)}{\partial u}\dot{u}$ instead of only $\dot{u}$ as in the controller design approach given later. $\frac{\partial g(u)}{\partial u}$ will tend to zero as $|u|$ is big enough as depicted in Fig. \ref{tanh and sat}, and if $\left(\frac{\partial g(u)}{\partial u}\right)^{-1}$ is used in the final designed control law, the singularity will appear. In order to avoid this situation, a Nussbaum function is employed.
\end{remark}

\begin{figure}
\centering
\epsfig{file=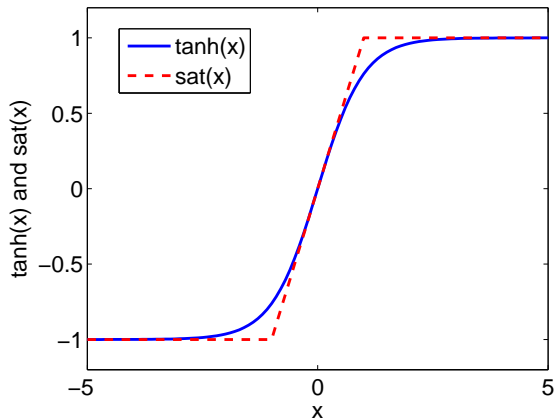,height=6cm} \caption{Comparison between $\tanh(\cdot)$ and $\mathrm{sat(\cdot)}$\label{tanh and sat}}
\end{figure}

\section{Guidance Law Design}\label{guidance law design}
Since $\dot{D}$ cannot be  measured, we will construct an output feedback controller  based on $x_0,x_1$.  Choose   sets of positive constants $h_1,h_2$ and $\alpha_{1},\alpha_{2},\alpha_{3}$ such that the matrices
\begin{align*}
F_{0}= \left[\begin{array}{cc}
-h_1 &1\\
-h_{2}&0\end{array}\right],~~~~~
A_{0}=\left[\begin{array}{ccc}
0 & 1 & 0\\
0 & 0 & 1\\
-\alpha_{3} & -\alpha_{2} & -\alpha_{1}\end{array}\right]
\end{align*}
are Hurwitz. Then we can introduce a high-gain observer
\begin{subequations}\label{observer}
\begin{equation}
\dot{\hat{x}}_{1}=\hat{x}_{2}+\frac{h_{1}}{\varepsilon}(x_{1}-\hat{x}_{1})
\end{equation}
\begin{equation}\label{hat_x2}
\dot{\hat{x}}_{2}=f(D,t)-\ddot{D}^{*}+g_{0}(D,t)g(u)+\frac{h_{2}}{\varepsilon^{2}}(x_{1}-\hat{x}_{1})
\end{equation}
\end{subequations}
and   an output  feedback virtual controller
\begin{alignat}{1}
g^{*}(u)=g_{0}^{-1}\left(-f+\ddot{D}^{*}-\frac{\alpha_3}{\varepsilon^{3}_{0}}x_{0}-\frac{\alpha_2}{\varepsilon^{2}_{0}}x_{1}-\frac{\alpha_1}{\varepsilon_{0}}\hat{x}_{2}\right)\label{output feedback virtual controller}
\end{alignat}
where $\varepsilon>0,\varepsilon_{0}>0$ are gain constants. For system (\ref{sample ISS system with sat}), the control law
\begin{subequations}\label{control law uc}
\begin{equation}
u_c=\tau N(\mathcal{X})\bar{u},~N(\mathcal{X})=e^{\mathcal{X}^2}\cos\left(\frac{\pi}{2}\mathcal{X}\right),~\dot{\mathcal{X}}=\gamma_{x}\bar{u}(g-g^*)
\end{equation}
\begin{equation}
\bar{u}=\frac{1}{\tau}\frac{\partial g}{\partial u}u-G(D,t)-k(g-g^*)
\end{equation}
\end{subequations}
where $k>0$, $\gamma_x>0$, and $G(D,t)=\frac{1}{2}g_0^2\widetilde{g}-\frac{\partial g^*}{\partial t}-\frac{\partial g^*}{\partial x_0}x_1-\frac{\partial g^*}{\partial x_1}p(D,t)+\frac{1}{2}\left\|\widetilde{g}\frac{\partial g^*}{\partial x_1}D\right\|^2-\frac{\partial g^*}{\partial \hat{x}_2}\dot{\hat{x}}_2$, can be designed, and the main results can be stated as the following theorem.
\begin{theorem}
Consider the closed-loop system composed of (\ref{sample ISS system with sat}), (\ref{observer}), (\ref{output feedback virtual controller}), and (\ref{control law uc}).  There exist positive constants $\varepsilon_0^{*}$ and $\varepsilon^{*}$ such that, for every $0<\varepsilon_0<\varepsilon_0^{*}$ and $0<\varepsilon<\varepsilon^{*}$, the state $x$ can converge to a small residual set around the origin.
\end{theorem}
\begin{proof}
The change of variables
\begin{equation}\label{transform in theorem 2}
\zeta_0=\frac{x_0}{\varepsilon_0^2},~\zeta_1=\frac{x_1}{\varepsilon_0},~\zeta_2=x_2,~\eta_{1}=\frac{x_{1}-\hat{x}_1}{\varepsilon},~\eta_{2}=x_{2}-\hat{x}_2
\end{equation}
bring Eqs. (\ref{sample ISS system for state1}) and (\ref{observer}) into the form
\begin{subequations}\label{compact singularly perturbed form}
\begin{equation}
\dot{\zeta}=\frac{1}{\varepsilon_0}A_0\zeta+B\left(\Delta_s+g_0\widetilde{g}+\frac{\alpha_1}{\varepsilon_{0}}\eta_{2}\right)
\end{equation}
\begin{equation}
\dot{\eta}=\frac{1}{\varepsilon}{F_{0}}\eta+{D_0}\Delta_s
\end{equation}
\end{subequations}
where $\widetilde{g}=g-g^*$, $\zeta=[\zeta_1,\zeta_2,\zeta_3]^{T},\eta=[\eta_{1},\eta_{2}]^{T}$, and $D_0= [0,1]^{T}$. Since $F_0$ and $A_0$ are Hurwitz matrices, there exist positive definite matrices  $P_0$ and $P$ satisfying  $P_0A_0+A_0^TP_0=-I$ and $PF_0+F_0^TP=-I$.  Then,  the derivative of Lyapunov function
\begin{equation}\label{Lyapunov}
V(\zeta,\eta)=\zeta^TP_0\zeta+\eta^TP\eta
\end{equation}
along the trajectories of system (\ref{compact singularly perturbed form})  is given by
\begin{alignat}{1}
\dot{V}=&-\frac{1}{\varepsilon_{0}}\|\zeta\|^2+2\zeta^TP_0B\left(\Delta_s+g_0\widetilde{g}+\frac{\alpha_1}{\varepsilon_{0}}\eta_{2}\right)-\frac{1}{\varepsilon}\|\eta\|^2+2\eta^TPD_0\Delta_s\nonumber\\
\leq& -\frac{1}{\varepsilon_{0}}\|\zeta\|^2+2\|\zeta\|\|P_0B\|\left(l\|D(\varepsilon_0)\|\|\zeta\|+d+|g_0\widetilde{g}|+\frac{\alpha_1}{\varepsilon_{0}}\|\eta\|\right)\nonumber\\
&-\frac{1}{\varepsilon}\|\eta\|^2+2\|\eta\|\|PD_0\|\left(l\|D(\varepsilon_0)\|\|\zeta\|+d\right)
\end{alignat}
where $D(\varepsilon_{0})=\mathrm{diag}[\varepsilon_{0}^{2},\varepsilon_{0},1] $. Considering Eq. (\ref{sample ISS system for state2}), we chose a Lyapunov function as
\begin{equation}
V_c(\zeta,\eta,g)=V(\zeta,\eta)+\frac{1}{2}(g-g^*)^2
\end{equation}
and its derivative is given by
\begin{alignat}{1}
\dot{V}_c(\zeta,\eta,g)
\leq& -\frac{1}{\varepsilon_{0}}\|\zeta\|^2+2\|\zeta\|\|P_0B\|\left(l\|D(\varepsilon_0)\|\|\zeta\|+d+|g_0\widetilde{g}|+\frac{\alpha_1}{\varepsilon_{0}}\|\eta\|\right)\nonumber\\
&-\frac{1}{\varepsilon}\|\eta\|^2+2\|\eta\|\|PD_0\|\left(l\|D(\varepsilon_0)\|\|\zeta\|+d\right)+\widetilde{g}\left(\frac{\partial g}{\partial u}\left(-\frac{1}{\tau}u+\frac{1}{\tau}u_{c}\right)-\dot{g}^*\right)\nonumber\\
=& -\frac{1}{\varepsilon_{0}}\|\zeta\|^2+2\|\zeta\|\|P_0B\|\left(l\|D(\varepsilon_0)\|\|\zeta\|+d+|g_0\widetilde{g}|+\frac{\alpha_1}{\varepsilon_{0}}\|\eta\|\right)\nonumber\\
&-\frac{1}{\varepsilon}\|\eta\|^2+2\|\eta\|\|PD_0\|\left(l\|D(\varepsilon_0)\|\|\zeta\|+d\right)\nonumber\\
&+\widetilde{g}\left(\frac{\partial g}{\partial u}\left(-\frac{1}{\tau}u+\frac{1}{\tau}u_{c}\right)-\frac{\partial g^*}{\partial t}-\frac{\partial g^*}{\partial x_0}x_1-\frac{\partial g^*}{\partial x_1}(p(D,t)+\delta D)-\frac{\partial g^*}{\partial \hat{x}_2}\dot{\hat{x}}_2\right)
\end{alignat}
Substituting the inequalities
$$
\|\zeta\|\|P_{0}B\|d\leq \frac{1}{2}\|\zeta\|^{2}+\frac{1}{2}\|P_{0}B\|^{2}d^{2}
$$$$
\|\eta\|\|PD_0\|d\leq \frac{1}{2}\|\eta\|^{2}+\frac{1}{2}\|PD_0\|^{2}d^{2}
$$
$$\|\zeta\|\|P_{0}B\||g_0\widetilde{g}|\leq \frac{1}{2}\|P_{0}B\|^{2}\|\zeta\|^{2}+\frac{1}{2}g_0^2\widetilde{g}^{2}$$
$$-\widetilde{g}\frac{\partial g^*}{\partial x_1}D\delta\leq \frac{1}{2}\left\|\widetilde{g}\frac{\partial g^*}{\partial x_1}D\right\|^2+\frac{1}{2}\delta^2$$
yields
\begin{alignat}{1}
\dot{V}_c
\leq&-\underbrace{\left(\frac{1}{\varepsilon_{0}}-1-2\|P_0B\|l\|D(\varepsilon_0)\|-\|P_{0}B\|^{2}\right)}_{\beta(\varepsilon_{0})}\|\zeta\|^2-\left(\frac{1}{\varepsilon}-1\right)\|\eta\|^2\nonumber\\
&+2\underbrace{\left(\|P_0B\|\frac{\alpha_1}{\varepsilon_{0}}+\|PD_0\|l\|D(\varepsilon_0)\|\right)}_{\alpha}\|\zeta\|\|\eta\|\nonumber\\
&+\widetilde{g}\left(\frac{\partial g}{\partial u}\left(-\frac{1}{\tau}u+\frac{1}{\tau}u_{c}\right)+\underbrace{\frac{1}{2}g_0^2\widetilde{g}-\frac{\partial g^*}{\partial t}-\frac{\partial g^*}{\partial x_0}x_1-\frac{\partial g^*}{\partial x_1}p(D,t)+\frac{1}{2}\left\|\widetilde{g}\frac{\partial g^*}{\partial x_1}D\right\|^2-\frac{\partial g^*}{\partial \hat{x}_2}\dot{\hat{x}}_2}_{G(D,t)}\right)\nonumber\\
&+\underbrace{(\|P_0B\|^2+\|PD_0\|^2)d_0^2+\frac{1}{2}\delta^2}_{\Xi}
\end{alignat}
Substituting (\ref{control law uc}) into the above inequality, one has
\begin{alignat}{1}
\dot{V}_c
\leq&-X^TQX+\Xi+\frac{\dot{\mathcal{X}}}{\gamma_x}\left(\frac{\partial g}{\partial u}N(\mathcal{X})-1\right)
\end{alignat}
where
\begin{equation}\nonumber
X=\begin{bmatrix}
\|\zeta\|\\
\|\eta\|\\
\widetilde{g}
\end{bmatrix},~
Q=\begin{bmatrix}
\beta(\varepsilon_{0}) & -\alpha & 0\\
-\alpha & \frac{1}{\varepsilon}-1 & 0\\
0 & 0 & k
\end{bmatrix}
\end{equation}
It can be seen that, for sufficiently small $\varepsilon$ and $\varepsilon_0$, one has $\beta(\varepsilon_{0})>0$ and the matrix $Q$ is positive define. Therefore, there exist positive constants $\varepsilon_0^{*}$ and $\varepsilon^{*}$ such that, for every $0<\varepsilon_0<\varepsilon_0^{*}$ and $0<\varepsilon<\varepsilon^{*}$, $\lambda_{\min}(Q)>0$ and the inequality
\begin{equation}\label{composite Lyapunov function}
\dot{V}_c\leq-\lambda_{\min}(Q)\|X\|^{2}+\Xi+\frac{\dot{\mathcal{X}}}{\gamma_x}\left(\frac{\partial g}{\partial u}N(\mathcal{X})-1\right)
\end{equation}
holds. Let $\mathcal{Y}=[\zeta,\eta,\widetilde{g}]^{T}$. Since $\|X\|=\|\mathcal{Y}\|$, we have
\begin{alignat}{1}
&\lambda_{\min}(P^{'})\|X\|^{2}=\lambda_{\min}(P^{'})\|\mathcal{Y}\|^{2}\leq V_c=\mathcal{Y}^{T}P^{'}\mathcal{Y}\leq \lambda_{\max}(P^{'})\|\mathcal{Y}\|^{2}=\lambda_{\max}(P^{'})\|X\|^{2}\label{inequality of V_(c)}
\end{alignat}
where $P^{'}=\mathrm{block~diag}\{P_{0},P,\frac{1}{2}\}$.
Substituting Eq. (\ref{inequality of V_(c)}) into Eq. (\ref{composite Lyapunov function}) yields
\begin{equation}
\dot{V}_c\leq -\frac{\lambda_{\min}(Q)}{\lambda_{\max}(P^{'})}V_c+\Xi+\frac{\dot{\mathcal{X}}}{\gamma_x}\left(\frac{\partial g}{\partial u}N(\mathcal{X})-1\right)
\end{equation}
that is
\begin{alignat}{1}
V_c(\mathcal{Y}(t))&\leq e^{-\lambda_{1} t}V_c(\mathcal{Y}(0))+\frac{1}{\lambda_{1}}(1-e^{-\lambda_{1} t})\Xi+\frac{1}{\gamma_x\lambda_{1}}\int_0^t\left(\frac{\partial g}{\partial u}N(\mathcal{X})\dot{\mathcal{X}}-\dot{\mathcal{X}}\right)e^{\lambda_1 (\tau-t)}\mathrm{d}\tau\nonumber\\
&\leq C + \frac{1}{\gamma_x\lambda_{1}}\int_0^t\left(\frac{\partial g}{\partial u}N(\mathcal{X})\dot{\mathcal{X}}-\dot{\mathcal{X}}\right)e^{\lambda_1 (\tau-t)}\mathrm{d}\tau\label{solution of dot(V_(c))}
\end{alignat}
where $\lambda_{1}=\frac{\lambda_{\min}(Q)}{\lambda_{\max}(P^{'})}$, $0<\frac{\partial g}{\partial u}=\frac{4}{(e^u+e^{-u})^2}\leq 1$, and $0<e^{-\lambda_{1} t}V_c(\mathcal{Y}(0))+\frac{1}{\lambda_{1}}(1-e^{-\lambda_{1} t})\Xi\leq C$ for bounded $\Xi$. It can easily be verified from Lemma \ref{Nussbaum lemma} and Eq. (\ref{solution of dot(V_(c))}) that, $\mathcal{X}(t)$ is bounded, and also, for a positive constant $\mathcal{X}_d$, we have $\left|\int_{0}^{t}\left(\frac{\partial g}{\partial u}N(\mathcal{X})\dot{\mathcal{X}}-\dot{\mathcal{X}}\right)e^{\lambda_1 (\tau-t)}\mathrm{d}\tau\right|
\leq \mathcal{X}_d$ . Substituting Eq. (\ref{inequality of V_(c)}) into Eq. (\ref{solution of dot(V_(c))}), we have
\begin{equation}\label{solution of dot(V_(c)) with inequality of V_(c)}
\|\mathcal{Y}(t)\|^{2}\leq \lambda_{2}e^{-\lambda_{1} t}\|\mathcal{Y}(0)\|^{2}+\frac{1}{\lambda_{1}\lambda_{3}}(1-e^{-\lambda_{1} t})\Xi+\frac{1}{\gamma_x\lambda_{1}\lambda_{3}}\mathcal{X}_d
\end{equation}
where $\lambda_{2}=\frac{\lambda_{\max}(P^{'})}{\lambda_{\min}(P^{'})},~\lambda_{3}=\lambda_{\min}(P^{'})$.
Therefore, $\mathcal{Y}(t)$ satisfies
\begin{equation}\label{state of closed-loop system in the similar form of ISS}
\|\mathcal{Y}(t)\|\leq \sqrt{\lambda_{2}e^{-\lambda_{1} t}}\|\mathcal{Y}(0)\|+\sqrt{\frac{1}{\lambda_{1}\lambda_{3}}(1-e^{-\lambda_{1} t})\Xi+\frac{1}{\gamma_x\lambda_{1}\lambda_{3}}\mathcal{X}_d}
\end{equation}
This indicates  the existence of the closed-loop  system solution. Due to the form of matrix $Q$ and $\lambda_{1}=\frac{\lambda_{\min}(Q)}{\lambda_{\max}(P^{'})}$, to guarantee that the state $x$
exponentially converges to a sufficient small residual set, we may choose the proper constants $\varepsilon_0^*$ and $\varepsilon^*$  such that $\sqrt{\frac{1}{\lambda_{1}\lambda_{3}}}$ is sufficiently small for every $0<\varepsilon_0<\varepsilon_0^{*}$ and $0<\varepsilon<\varepsilon^{*}$.
\end{proof}

\begin{figure}[H]
  \centering
  \subfigure[Drag acceleration]{
    \includegraphics[width=2.5in]{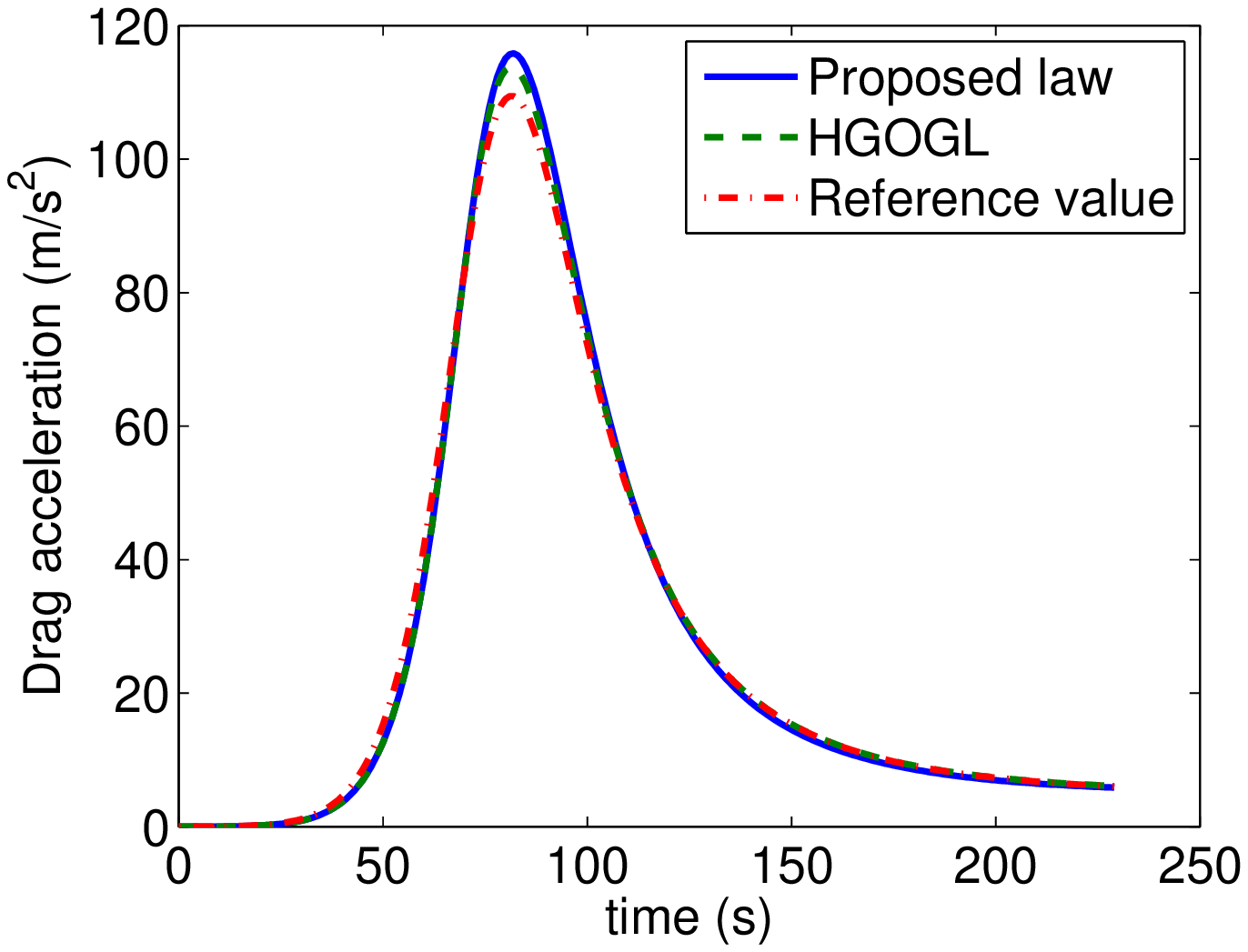}}
  \subfigure[Velocity]{
    \includegraphics[width=2.5in]{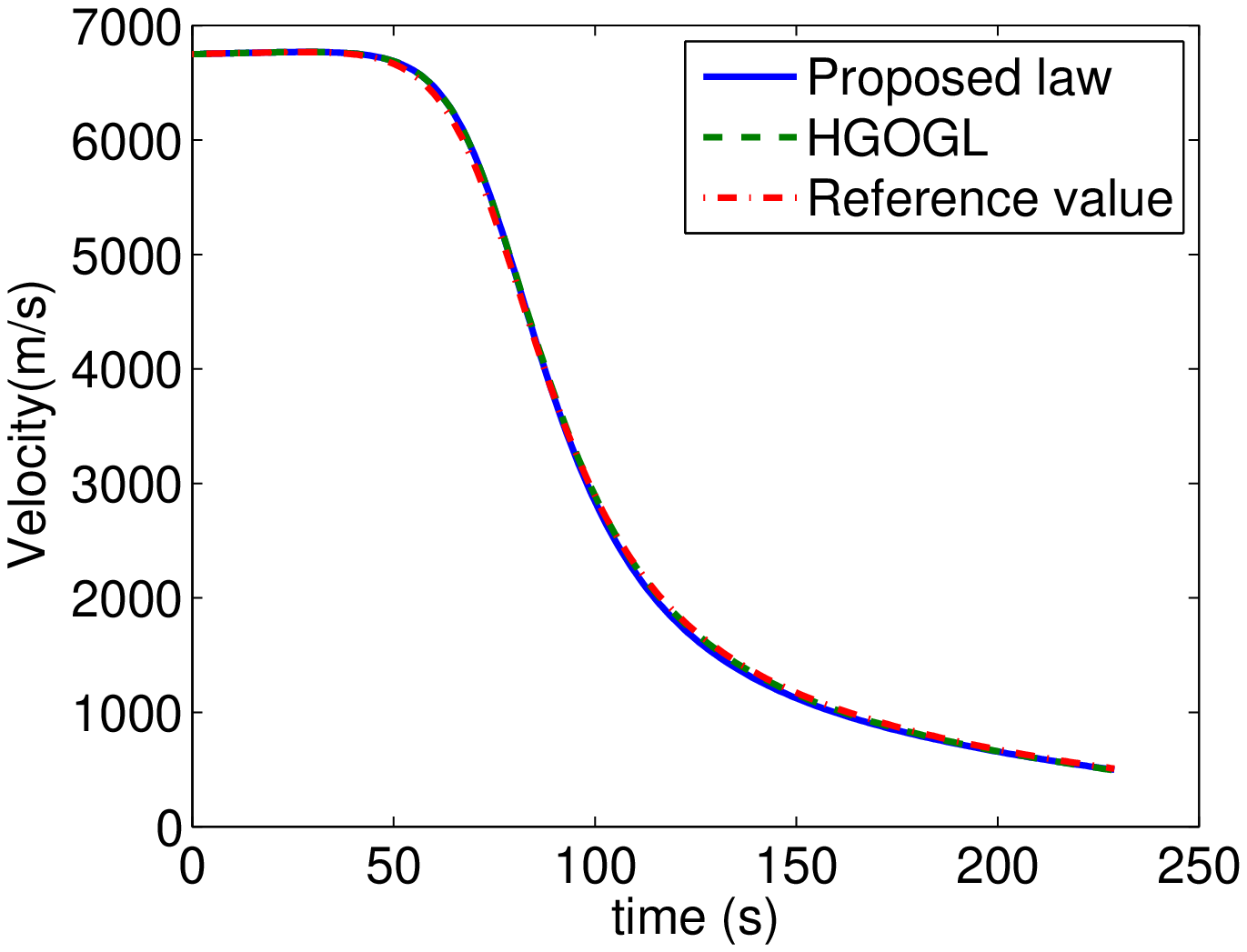}}
  \caption{Drag and velocity\label{drag}}
\end{figure}

\section{Simulation Results}\label{simulation}
This section presents simulation results to test the performance of the proposed guidance laws.

Consider the Mars atmospheric entry flight, and vehicle, reference drag profile and other data from \cite{Guo_2013} are used. The lift-to-drag ratio and the ballistic coefficient are 0.18 and $115\mathrm{kg/m^{2}}$, respectively. The initial and final state variables can be found in Table \ref{State}. It can be calculated out that the desired total downrange is 723.32km. The guidance command is saturated as $0^{\circ}\leq\sigma\leq180^{\circ}$.
\begin{figure}[H]
  \centering
  \subfigure[Altitude]{
    \includegraphics[width=2.5in]{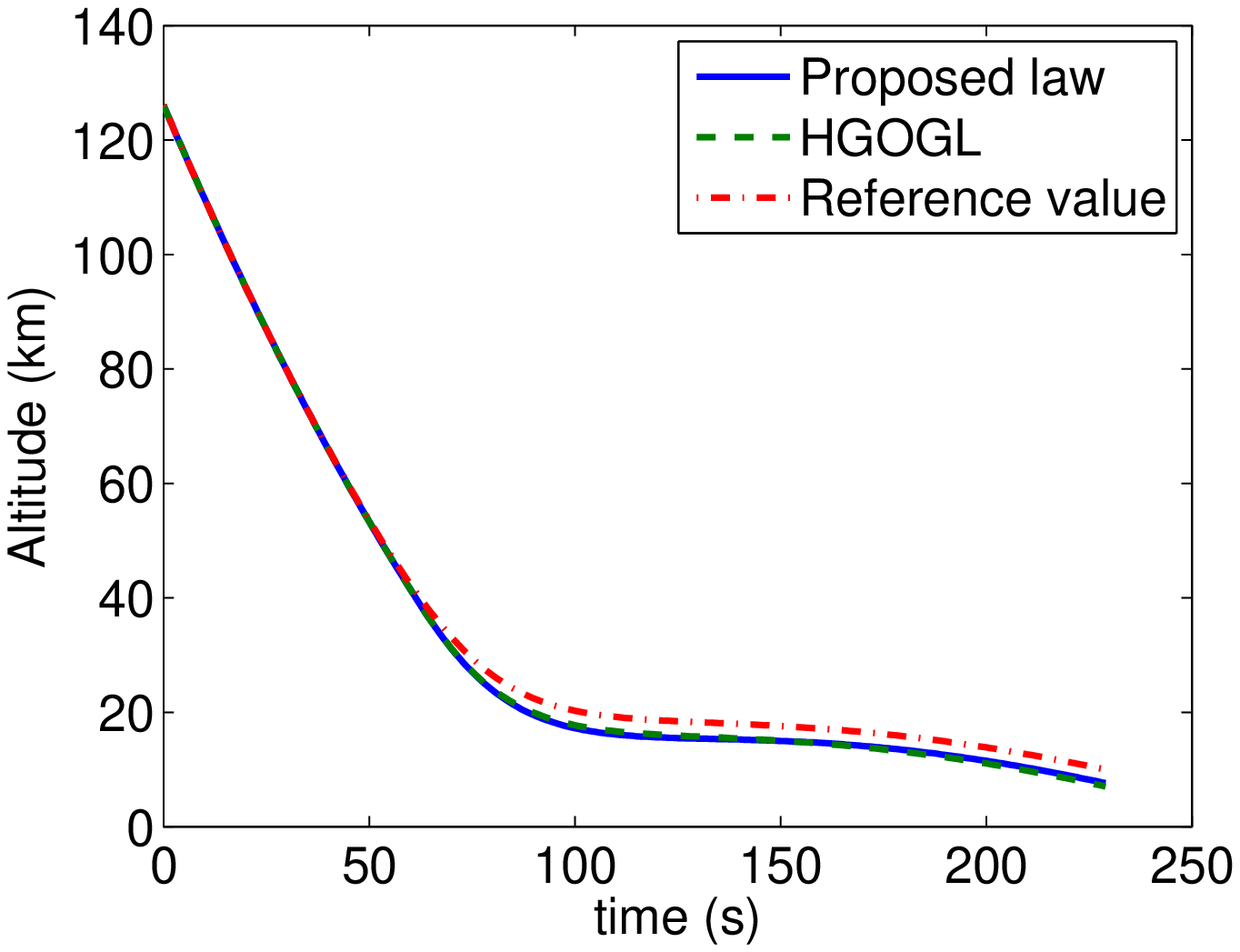}}
  \subfigure[Downrange error]{
    \includegraphics[width=2.5in]{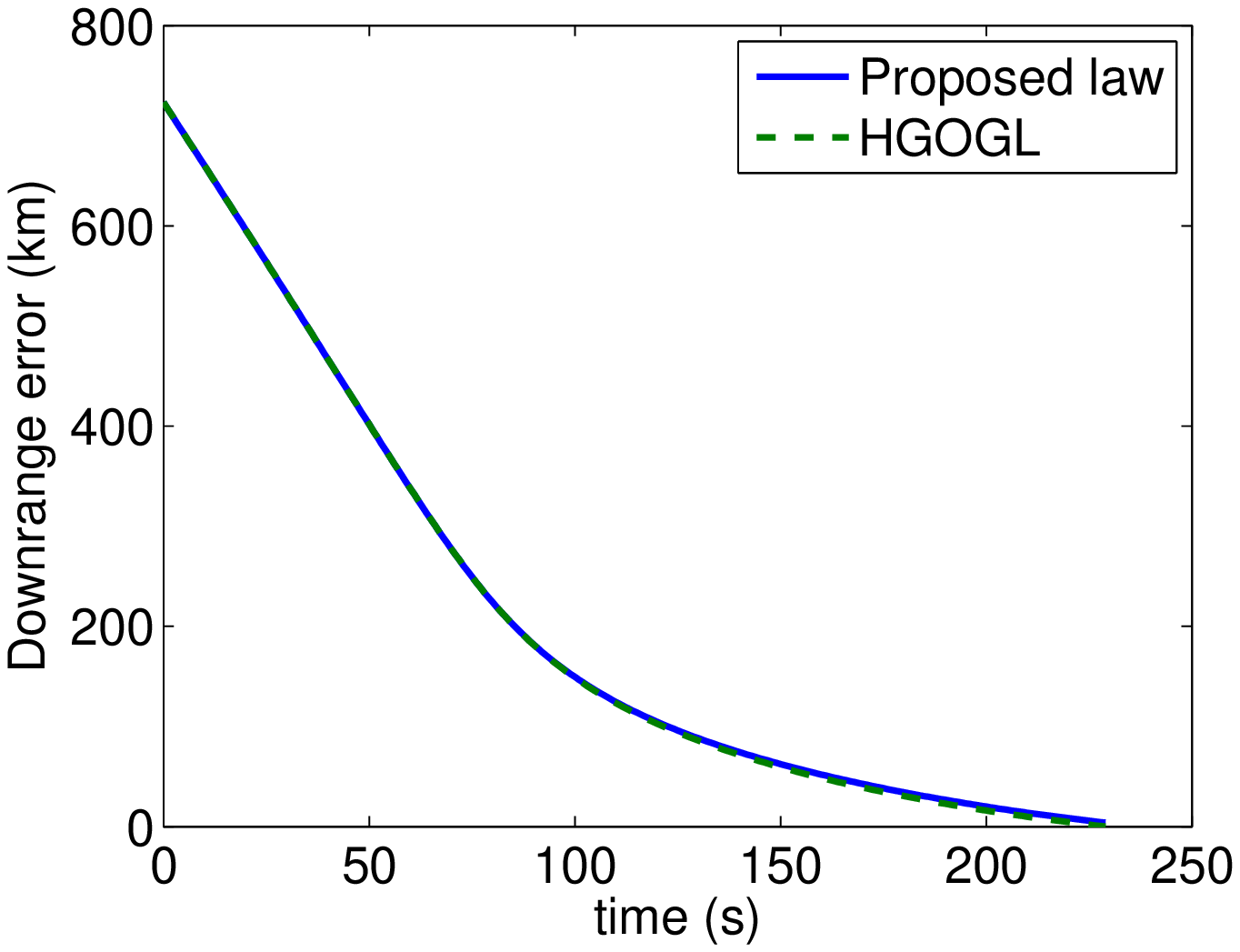}}
  \caption{Altitude and downrange error\label{Altitude}}
\end{figure}

\begin{figure}[H]
  \centering
  \subfigure[$\sigma$]{
    \includegraphics[width=2.5in]{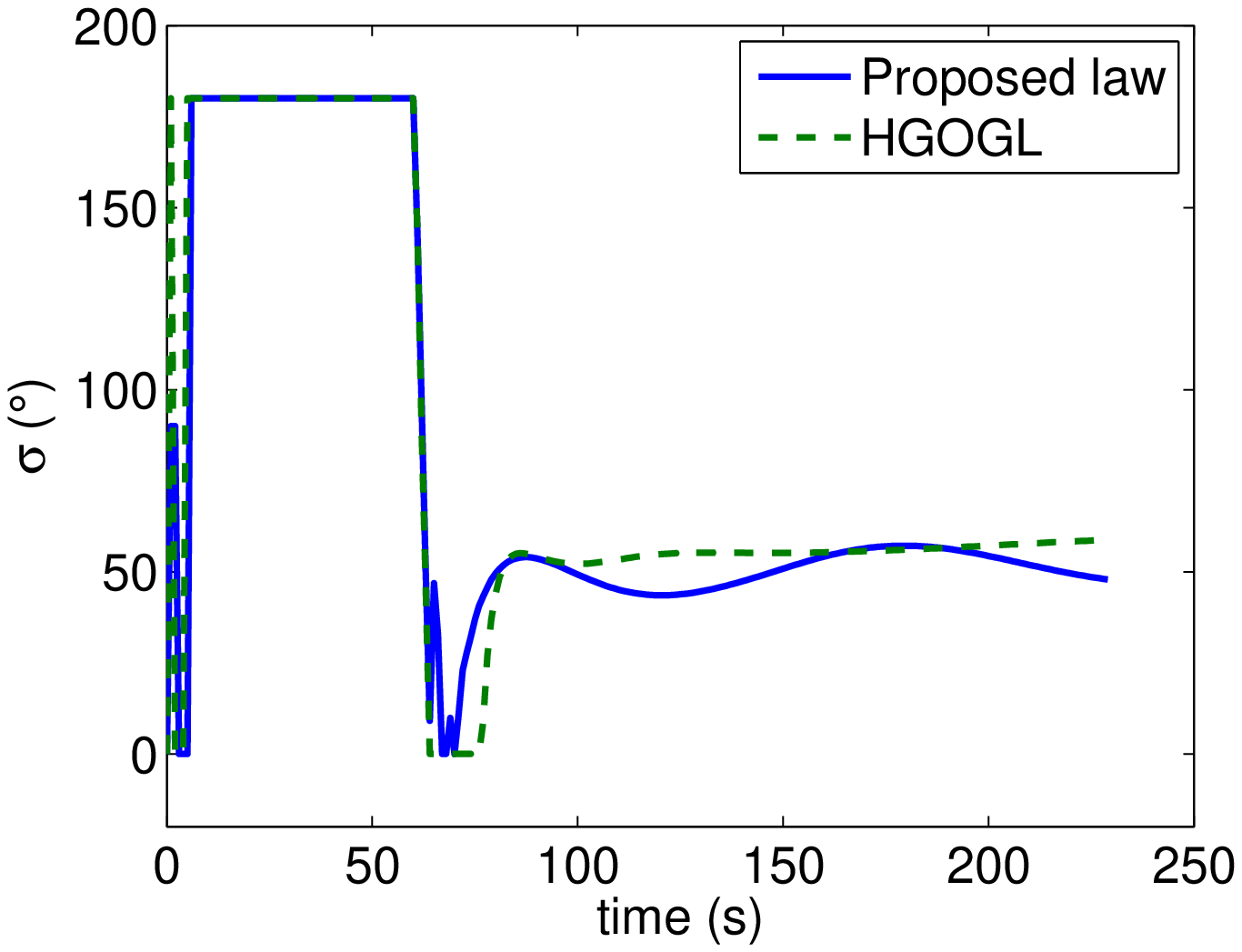}}
  \subfigure[$\gamma$]{
    \includegraphics[width=2.5in]{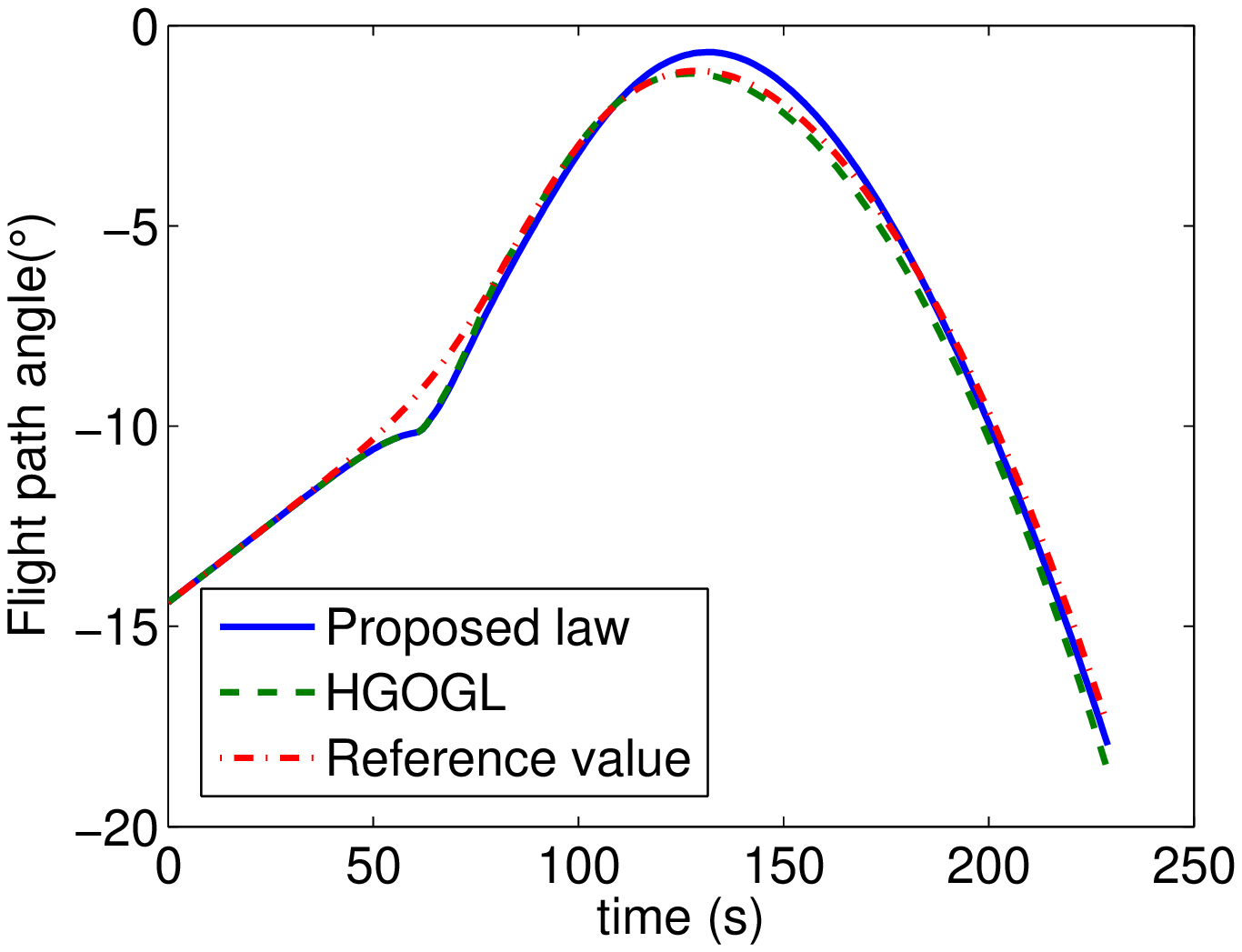}}
  \caption{Bank angle and flight path angel\label{bankandsigma}}
\end{figure}

\begin{figure}[H]
  \centering
  \subfigure[Drag rate with the proposed law]{
    \includegraphics[width=2.5in]{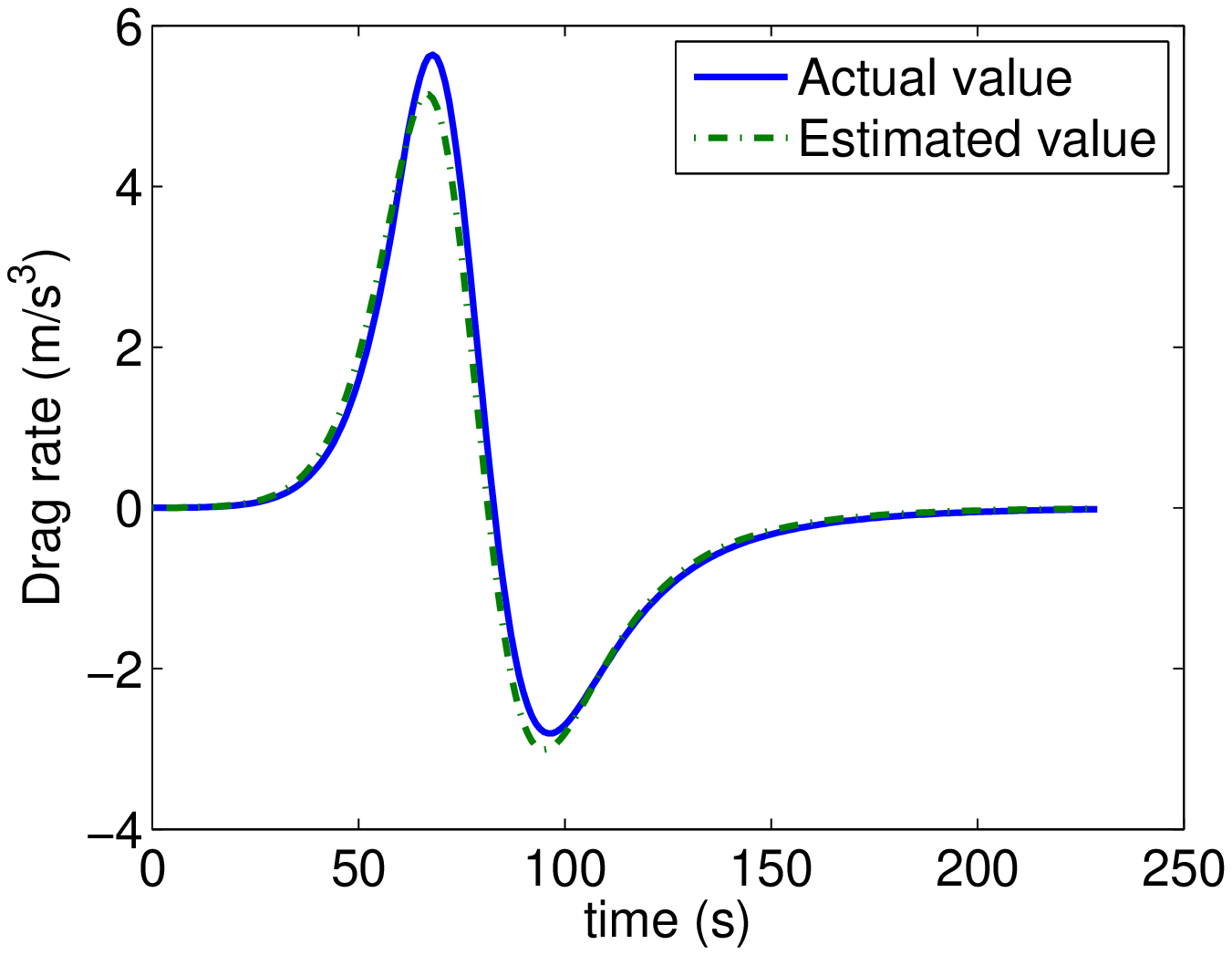}}
  \subfigure[Drag rate estimate error]{
    \includegraphics[width=2.5in]{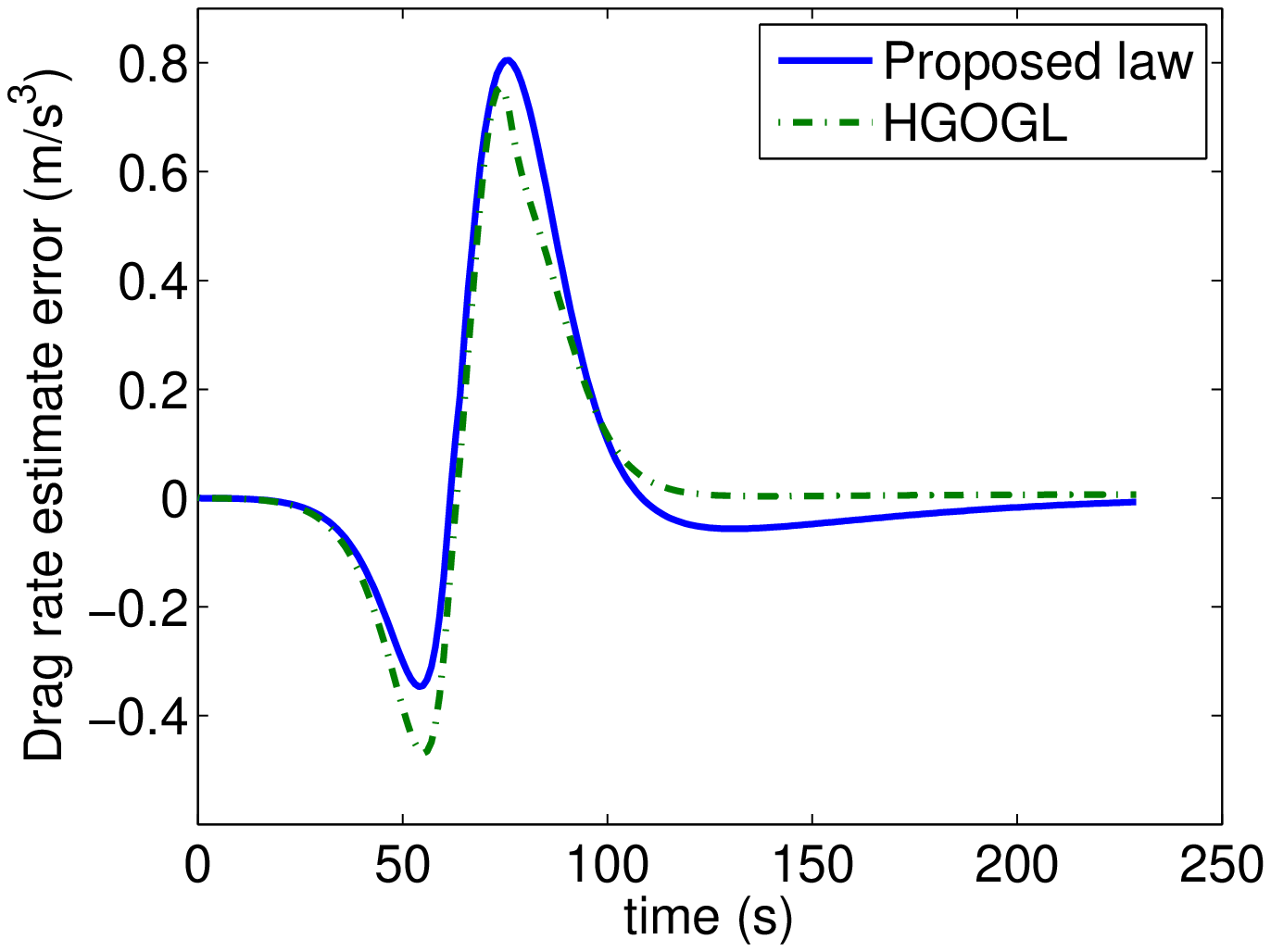}}
  \caption{Drag rate and its estimate error\label{drag_rate}}
\end{figure}

\begin{table}[h]
\begin{center}
\caption{State Variables\label{State}}
\begin{tabular}{c|c}\hline\hline
Initial State Variables & Final State Variables \\ \hline
\begin{tabular}{c|c}
Altitude, $h_0$ (km) & 126.1\\
Relative velocity, $V_0$ (km/s) & 6.75\\
Flight path angle,  $\gamma_{0}$ ($\circ$) & -14.4\\
Longitude ($\circ$) & 0\\
Latitude ($\circ$) & 0
\end{tabular}&
\begin{tabular}{c|c}
Altitude, $h_f$ (km) & 10 \\
Relative Velocity, $V_f$ (m/s) & 503 \\
Flight path angle,  $\gamma_{f}$ ($\circ$) & --- \\
Longitude ($\circ$) & 12.2 \\
Latitude ($\circ$) & 0
\end{tabular} \\ \hline\hline
\end{tabular}
\end{center}
\end{table}

In order to compare the proposed guidance law with the high-gain observer based guidance law (HGOGL) in \cite{Yan}, the performance of the two laws are shown in Figs. \ref{drag}-\ref{drag_rate}. The parameters for the proposed guidance law are taken as $a=1.982,b=3,c=0.1,\varepsilon_0=6.5,h_1=2h_2=2,\varepsilon=1.78,\tau=1,\gamma_x=0.0005,k=1$, and the parameters for HGOGL are the same as \cite{Yan}. We can see that the laws have similarity performance in this case. Since the atmospheric density is very small at the beginning of entry and it leads to the fact that $g_{0}(D,t)=-\left(\frac{v}{h_{s}}+\frac{2g}{v}\right)\frac{LD\cos\gamma}{v}$ is small, thus, a large control magnitude is needed to make the drag track its reference value, which is the reason why bank angle reaches saturation level at initial time with both guidance laws. But it can be seen from Fig. \ref{bankandsigma} that, after about 60s, under the law proposed in this paper, the duration of bank angle reaching the saturation level is much shorter then HGOGL.

\begin{figure}[H]
  \centering
    \subfigure[Downrange error]{
    \includegraphics[width=2.5in]{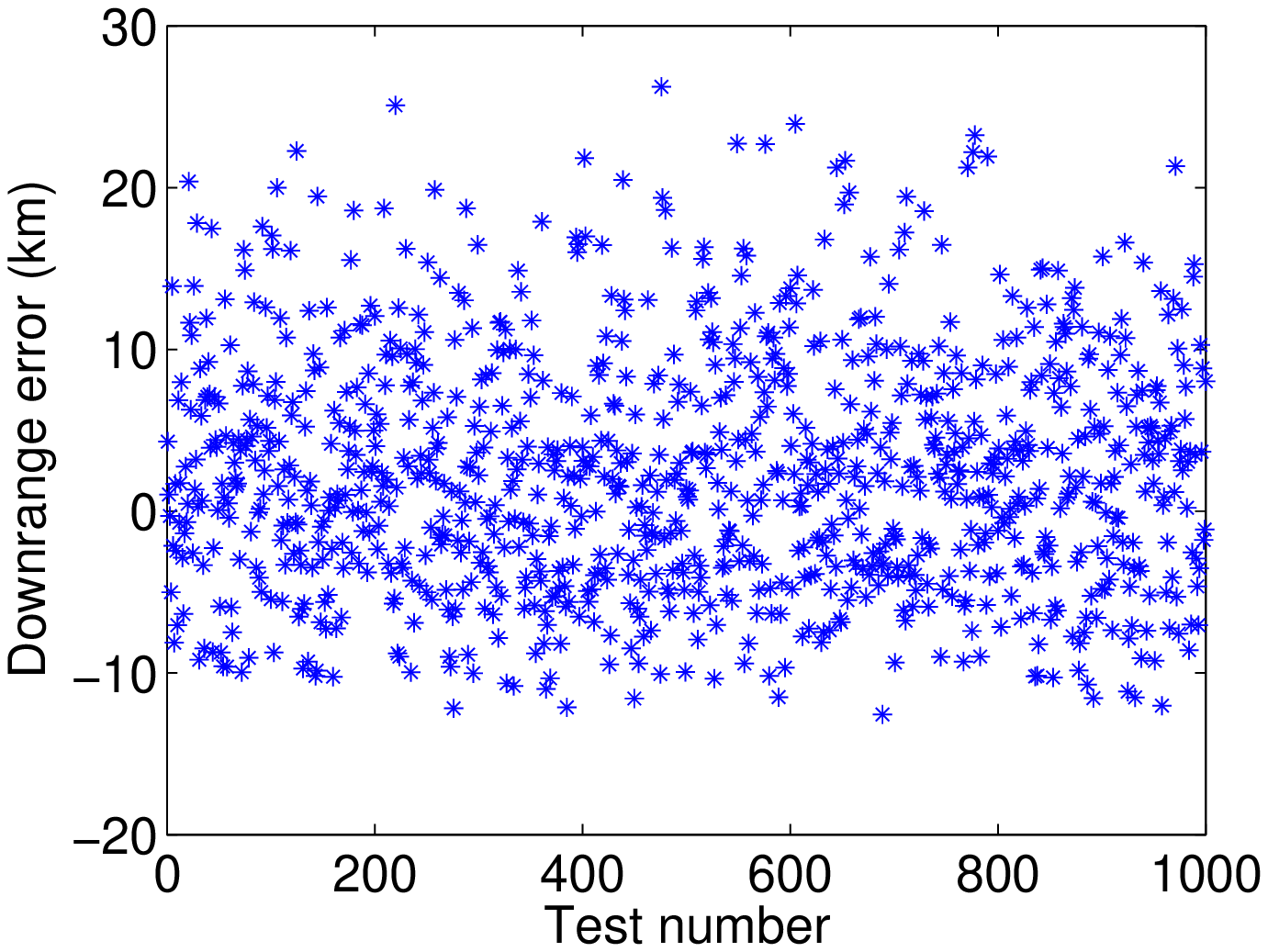}}
    \subfigure[Altitude error]{
    \includegraphics[width=2.5in]{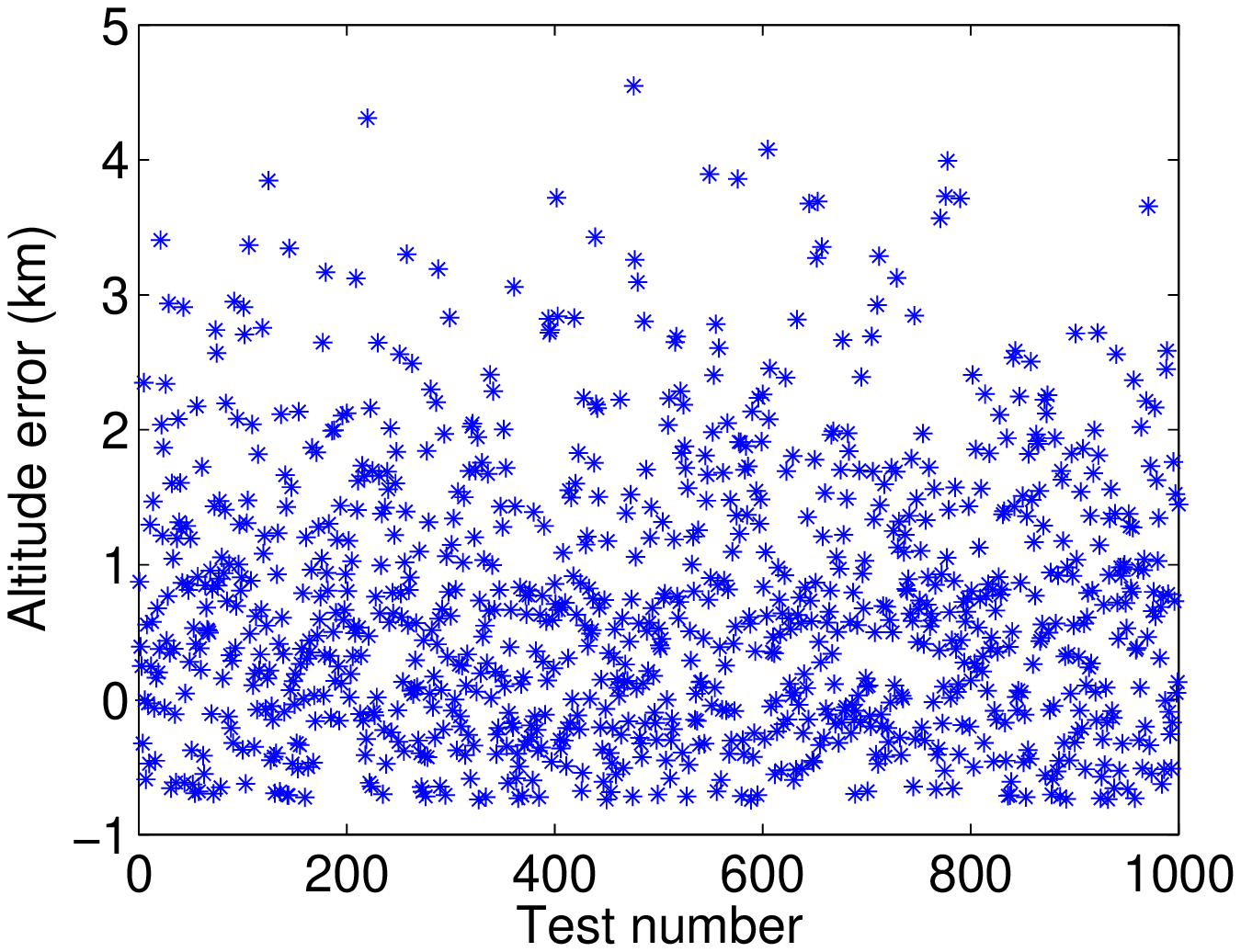}}
  \caption{Monte Carlo study of guidance law in \cite{Yan}\label{daba_duibi}}
\end{figure}

\begin{figure}[H]
  \centering
    \subfigure[Downrange error]{
    \includegraphics[width=2.5in]{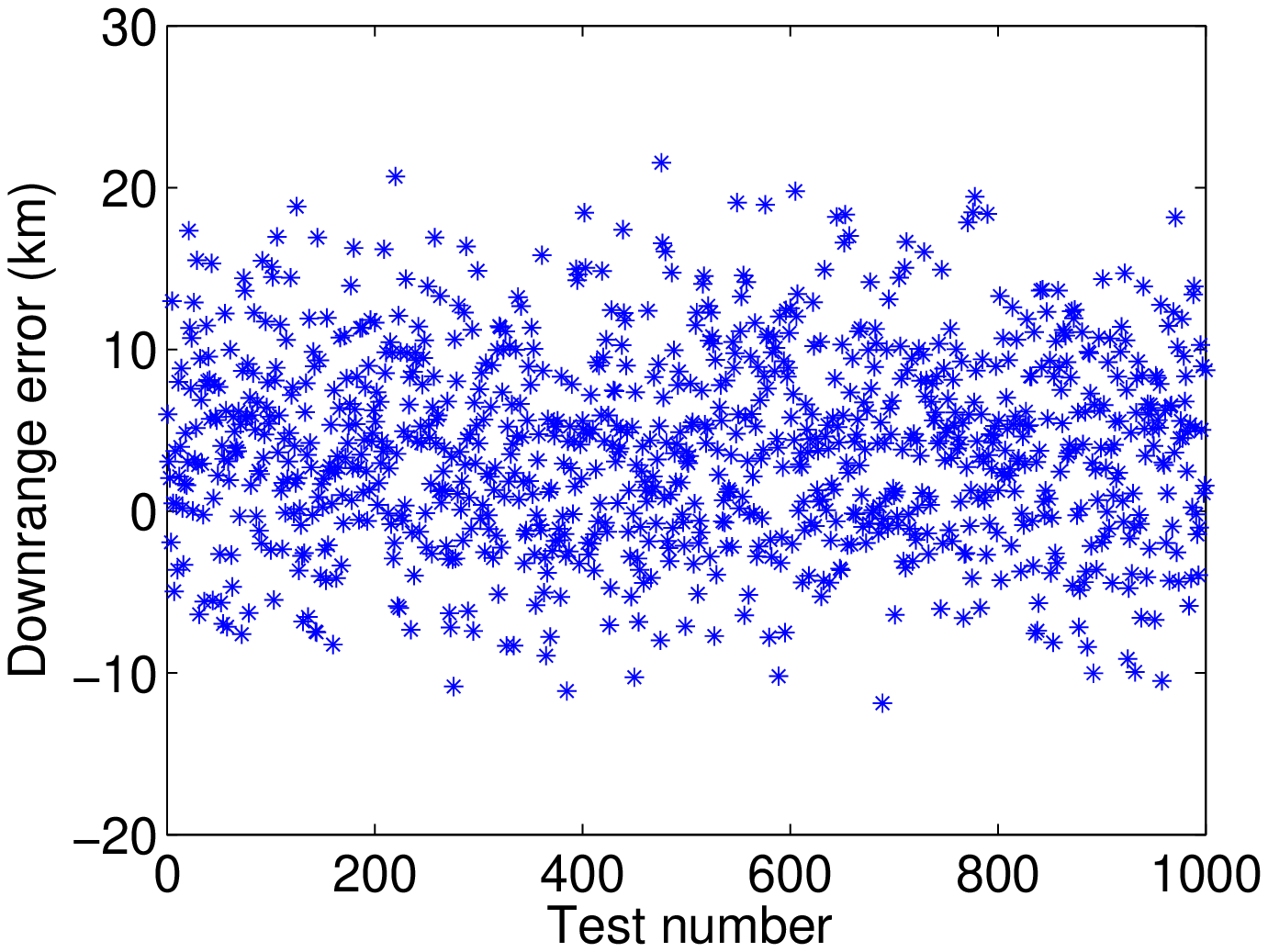}}
    \subfigure[Altitude error]{
    \includegraphics[width=2.5in]{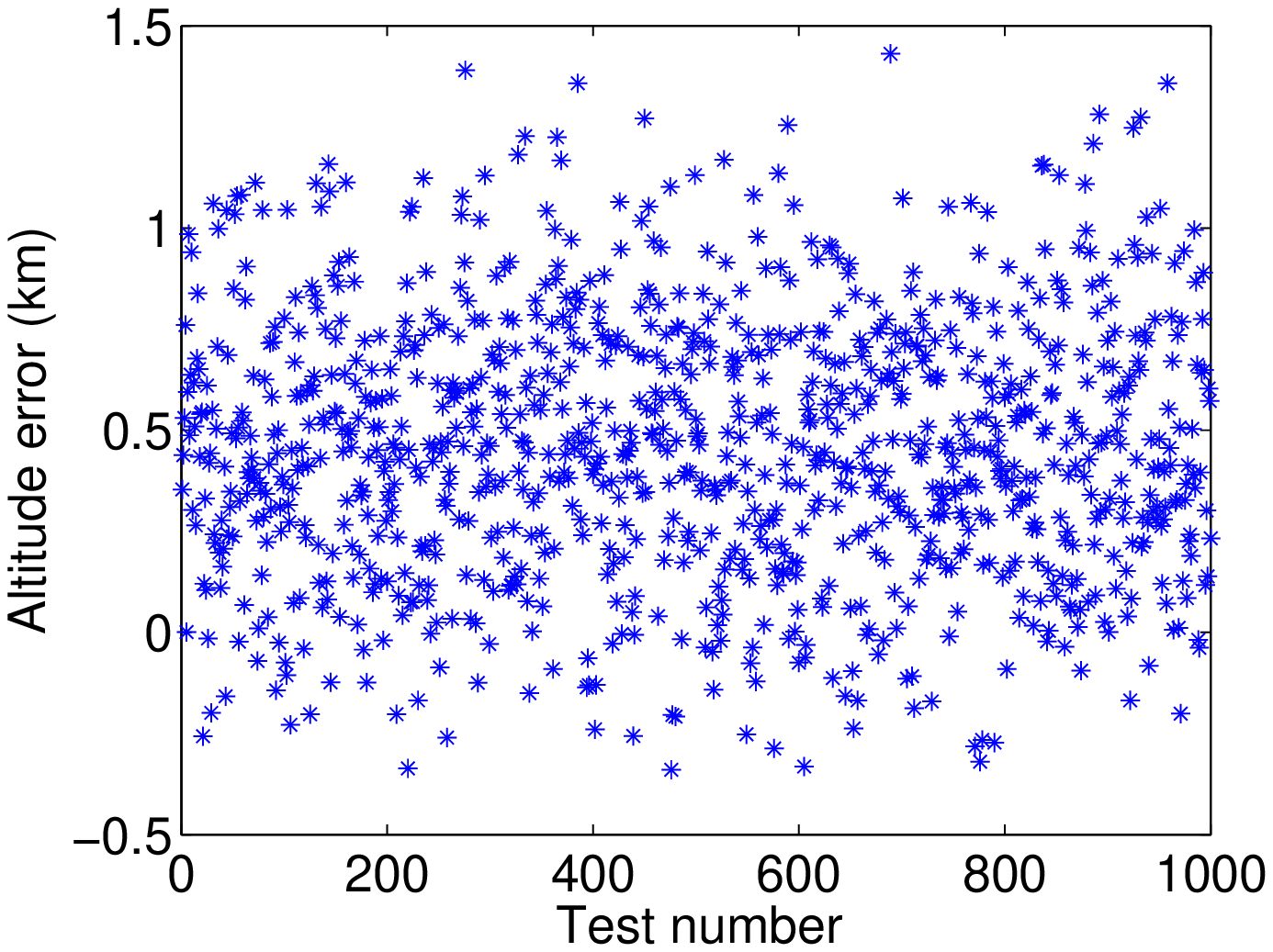}}
  \caption{Monte Carlo study of the proposed guidance law\label{1000daba}}
\end{figure}

\begin{table}[h]
\begin{center}
\caption{Statistics of Dispersions Used in Monte Carlo Study\label{daba}}
\begin{tabular}{c|c|c}\hline\hline
Parameters & Distribution & $[\Delta^{-},\Delta^{+}]$ \\ \hline
Mass deviation & uniform & [-5\%,5\%]\\
Atmospheric density deviation & uniform & [-20\%,20\%]\\
$C_{L}$ deviation & uniform & [-30\%,30\%]\\
$C_{D}$ deviation & uniform & [-30\%,30\%]\\ \hline\hline
\end{tabular}
\end{center}
\end{table}

A 1000-run Monte Carlo study using the parameter deviation in Table \ref{daba} is also done to test the robustness of the proposed guidance law. Take $a=1.982,b=3,c=0.1,\varepsilon_0=6.5,h_1=2h_2=2,\varepsilon=0.425,\tau=1,\gamma_x=0.0005,k=1$, and the result is shown in Fig. \ref{1000daba}. For comparison with the existing work, the Monte Carlo simulation result of the guidance law in \cite{Yan} is also depicted, and the statistical results of these Monte Carlo simulation are summarized in Table \ref{daba result}. Comparing with \cite{Yan}, the integral term is introduced in the proposed guidance law, and it can further eliminate the uncertainties as we analyzed. Thus, we can obviously see from both Figs. \ref{daba_duibi}-\ref{1000daba} and Table \ref{daba result} that the downrange and altitude errors are more close to zero by using the law proposed in this paper.

\begin{table}[h]
\begin{center}
\caption{Result of Monte Carlo Study\label{daba result}}
\begin{tabular}{c|c|c|c|c}\hline\hline
\multirow{2}*{}& \multicolumn{2}{|c|}{Downrange Error (km)} & \multicolumn{2}{|c}{Altitude Error (km)}\\ \cline{2-5}
&Proposed law& The law in \cite{Yan} &Proposed law& The law in \cite{Yan}\\\hline
Minimum & 0.0141 & -0.0191 & 0.000682 & 0.000311\\
Maximum & 21.5357 & 26.2333 & 1.4304 & 4.5474\\
Average & 4.0354 & 2.4492 & 0.4576 & 0.7056\\
Standard deviation & 5.9681 & 7.3655 &  0.3247 & 1.0129\\ \hline\hline
\end{tabular}
\end{center}
\end{table}

\section{Conclusions}\label{conclusion}
This paper proposes a robust output feedback drag-tracking guidance law for entry vehicles in the presence of input saturation. By employing a Nussbaum type function to deal with the problem of input saturation constraint, an output feedback guidance law with high-gain observer is proposed. Comparing with previous work, an integral term is introduced to get a better result, and the stability analysis is also given by considering the input saturation. The simulation results show the advantage in both altitude and downrange control when using the designed law.

\section*{APPENDIX}
\section*{The Proof of Lemma \ref{Nussbaum lemma}}
Define
\begin{equation}\label{Vx}
V_{\mathcal{X}}(t_i,t_j)=\int_{t_i}^{t_j}\left(\rho(\tau) N(\mathcal{X})\dot{\mathcal{X}}-\dot{\mathcal{X}}\right)e^{\lambda (\tau-t_j)}\mathrm{d}\tau
\end{equation}
Due to the fact that
$$0<\rho_{\min}\leq \rho(t)\leq \rho_{\max}$$
and
$$0<e^{\lambda (\tau-t_j)}\leq1$$
for $\tau\in[t_i,t_j]$. Thus
\begin{alignat}{1}
|V_{\mathcal{X}}(t_i,t_j)|\leq& \int_{\mathcal{X}(t_i)}^{\mathcal{X}(t_j)}(\rho_{\max}|N(\mathcal{X})|+1)\mathrm{d}\mathcal{X}\nonumber\\
\leq &(\mathcal{X}(t_j)-\mathcal{X}(t_i))\left(\rho_{\max}\sup_{x\in[\mathcal{X}(t_i),\mathcal{X}(t_j)]}|N(x)|+1\right)\label{Vx_bound}
\end{alignat}
We will seek a contradiction to show that $V_{\mathcal{X}}(t_i,t_j)$ and $\mathcal{X}(t)$ are bounded. It is supposed that $\mathcal{X}(t)$ is unbounded and two cases should be considered: (1) $\mathcal{X}(t)$ has no upper bound, and (2) $\mathcal{X}(t)$ has no lower bound.

Case 1: Suppose that $\mathcal{X}(t)$ has no upper bound on $[0,t_f]$. There must exist a monotone increasing variable $\mathcal{X}$ with $\mathcal{X}_0=\mathcal{X}(t_0)>0,\lim_{i\rightarrow\infty}t_i=t_f,\lim_{i\rightarrow\infty}\mathcal{X}_i=\infty$. From Eq. (\ref{Vx_bound}), for the interval $[\mathcal{X}(t_0),\mathcal{X}(t_1)]=[\mathcal{X}_0,4m+1]$ with positive integer $m$,
\begin{alignat}{1}
V_{\mathcal{X}}(t_0,t_1)
\leq &(\mathcal{X}(t_1)-\mathcal{X}_0)\left(\rho_{\max}\sup_{x\in[\mathcal{X}(t_0),\mathcal{X}(t_1)]}|N(x)|+1\right)\nonumber\\
=&\underbrace{(4m+1-\mathcal{X}_0)}_{c_{x1}}(\rho_{\max}e^{{(4m+1)}^2}+1)\label{V01}
\end{alignat}
Noting the facts that, for $\forall x\in[\mathcal{X}(t_1),\mathcal{X}(t_2)]=(4m+1,4m+3]$ and $\forall \tau\in[t_1,t_2]$, we have $N(x)\leq0$ and $e^{\lambda_1 (\tau-t_2)}\geq e^{\lambda_1 (t_1-t_2)}$, it can be obtained that
\begin{alignat}{1}
V_{\mathcal{X}}(t_1,t_2)
&\leq \int_{4m+2-\varepsilon}^{4m+2+\varepsilon}\left(\rho_{\min} N(\mathcal{X})-1\right)e^{\lambda_1 (\tau-t_2)}\mathrm{d}\mathcal{X}\nonumber\\
&\leq 2\varepsilon\left(\sup_{x\in[4m+2-\varepsilon,4m+2+\varepsilon]}N(x)\rho_{\min}-1\right)e^{\lambda_1 (t_1-t_2)}
\end{alignat}
where $\varepsilon\in(0,1)$. Since $N(\mathcal{X})=e^{\mathcal{X}^2}\cos\left(\frac{\pi}{2}\mathcal{X}\right)$, it can be easily to verify that
\begin{alignat}{1}
V_{\mathcal{X}}(t_1,t_2)
&\leq 2\varepsilon\left(\rho_{\min}N(4m+2-\varepsilon)-1\right)e^{\lambda_1 (t_1-t_2)}\nonumber\\
&=2\varepsilon\left(-\rho_{\min}e^{(4m+2-\varepsilon)^2}\cos\left(\frac{\pi}{2}\varepsilon\right)-1\right)e^{\lambda_1 (t_1-t_2)}\nonumber\\
&=-e^{(4m+2-\varepsilon)^2}\underbrace{2\rho_{\min}\varepsilon\cos\left(\frac{\pi}{2}\varepsilon\right)e^{\lambda_1 (t_1-t_2)}}_{c_{x2}}-\underbrace{2\varepsilon e^{\lambda_1 (t_1-t_2)}}_{c_{x3}}\label{V12}
\end{alignat}
From Eqs. (\ref{V01}) and (\ref{V12}), we have
\begin{alignat}{1}
V_{\mathcal{X}}(t_0,t_2)
&=V_{\mathcal{X}}(t_0,t_1)+V_{\mathcal{X}}(t_1,t_2)\nonumber\\
&\leq-e^{(4m+1)^2}(c_{x2}e^{(8m+3-\varepsilon)(1-\varepsilon)}-c_{x1}\rho_{\max})+c_{x1}-c_{x3}
\end{alignat}
From above equation, it is known that $V_{\mathcal{X}}(t_0,t_2)=V_{\mathcal{X}}(t_0,4m+3)\rightarrow-\infty$ as $m\rightarrow+\infty$, that is, from Eq. (\ref{condition for lemma1}), $V<0$ at this time. On the other hand, $V\geq0$ for $\forall t$. Thus we can always find a subsequence that leads to a contradiction. Therefore, $\mathcal{X}(t)$ has a upper bound.

Case 2: Suppose that $\mathcal{X}(t)$ has no lower bound on $[0,t_f]$. Define $\mathcal{X}(t)=-\nu(t)$, and then, $\nu(t)$ has no upper bound. Consider that $N(\cdot)$ is an even function, thus
\begin{equation}\label{solution of dot(V_(c))2}
0\leq V(t)\leq c_0 - c_1\int_0^t\left(\rho(\tau)N(\nu(\tau))\dot{\nu}(\tau)-\dot{\nu}(\tau)\right)e^{\lambda (\tau-t)}\mathrm{d}\tau
\end{equation}
Define
\begin{equation}
V_{\nu}(t_i,t_j)=\int_{t_i}^{t_j}\left(\rho(\tau)N(\nu)\dot{\nu}-\dot{\nu}\right)e^{\lambda (\tau-t_j)}\mathrm{d}\tau
\end{equation}
There must exist a monotone increasing variable $\nu$ with $\nu_0=\nu(t_0)>0,\lim_{i\rightarrow\infty}t_i=t_f,\lim_{i\rightarrow\infty}\nu_i=\infty$, and following the similar logical deduction as Eq. (\ref{Vx})-(\ref{V01}), for the interval $[\nu(t_0),\nu(t_1)]=[\nu_0,4m-1]$ with positive integer $m$, we have
\begin{alignat}{1}
V_{\nu}(t_0,t_1)
\geq &-(\nu(t_1)-\nu_0)\left(\rho_{\max}\sup_{x\in[\nu(t_0),\nu(t_1)]}|N(x)|+1\right)\nonumber\\
=&-\underbrace{(4m-1-\nu_0)}_{c_{\nu1}}(\rho_{\max}e^{{(4m-1)}^2}+1)\label{V012}
\end{alignat}
Noting the facts that, for $\forall x\in[\nu(t_1),\nu(t_2)]=(4m-1,4m+1]$ and $\forall \tau\in[t_1,t_2]$, we have $N(x)\geq0$ and $0>e^{\lambda_1 (\tau-t_2)}\geq e^{\lambda_1 (t_1-t_2)}$, it can be obtained that
\begin{alignat}{1}
V_{\nu}(t_1,t_2)
&\geq \int_{4m-\varepsilon}^{4m+\varepsilon}\left(\rho_{\min}N(\nu)-1\right)e^{\lambda_1 (\tau-t_2)}\mathrm{d}\nu\nonumber\\
&\geq 2\varepsilon\rho_{\min}e^{\lambda_1 (t_1-t_2)}\inf_{x\in[4m-\varepsilon,4m+\varepsilon]}N(x)-2\varepsilon
\end{alignat}
where $\varepsilon\in(0,1)$. Since $N(\mathcal{X})=e^{\mathcal{X}^2}\cos\left(\frac{\pi}{2}\mathcal{X}\right)$, it can be easily to verify that
\begin{alignat}{1}
V_{\nu}(t_1,t_2)
&\geq 2\varepsilon\rho_{\min}e^{\lambda_1 (t_1-t_2)}N(4m-\varepsilon)-2\varepsilon\nonumber\\
&=e^{(4m-\varepsilon)^2}\underbrace{2\rho_{\min}\varepsilon\cos\left(\frac{\pi}{2}\varepsilon\right)e^{\lambda_1 (t_1-t_2)}}_{c_{\nu2}}-\underbrace{2\varepsilon}_{c_{\nu3}}\label{V122}
\end{alignat}
From Eqs. (\ref{V012}) and (\ref{V122}), we have
\begin{alignat}{1}
V_{\nu}(t_0,t_2)
&=V_{\nu}(t_0,t_1)+V_{\nu}(t_1,t_2)\nonumber\\
&\geq e^{(4m-1)^2}(c_{\nu2}e^{(8m-1-\varepsilon)(1-\varepsilon)}-\rho_{\max}c_{\nu1})-c_{\nu1}-c_{\nu3}
\end{alignat}
From above equation, it is known that $V_{\nu}(t_0,t_2)=V_{\mathcal{X}}(t_0,4m+1)\rightarrow+\infty$ as $m\rightarrow+\infty$, that is, from Eq. (\ref{solution of dot(V_(c))2}), $V<0$ at this time. On the other hand, $V\geq0$ for $\forall t$. Thus we can always find a subsequence that leads to a contradiction. Therefore, $\mathcal{X}(t)$ has a lower bound.

Accordingly, we can conclude that $V(t)$ and $\int_0^t\left(\rho(\tau)N(\mathcal{X})\dot{\mathcal{X}}-\dot{\mathcal{X}}\right)e^{\lambda (\tau-t)}\mathrm{d}\tau$ are also bounded due to Eqs. (\ref{Vx_bound}) and (\ref{condition for lemma1}).


\begin{thebibliography}{1}
\bibitem{Tian_2011}
B. L. Tian and Q. Zong, Optimal guidance for reentry vehicles based on indirect Legendre pseudospectral method, Acta Astronautica, Vol. 68, No. 7-8, 2011, pp. 1176-1184.
\bibitem{ISA_2014}
X. L. Shao and H. L. Wang, Sliding mode based trajectory linearization control for hypersonic reentry vehicle via extended disturbance observer, ISA Transactions, Vol. 53, No. 6, 2014, pp. 1771-1786.
\bibitem{ISA_2015}
X. L. Shao and H. L. Wang, Active disturbance rejection based trajectory linearization control for hypersonic reentry vehicle with bounded uncertainties, ISA Transactions, Vol. 54, 2015, pp. 27-38.
\bibitem{AESCTE_2015}
X. L. Shao, H. L. Wang and H. P. Zhang, Enhanced trajectory linearization control based advanced guidance and control for hypersonic reentry vehicle with multiple disturbances, Aerospace Science and Technology, Vol. 46, 2015, pp. 523-536.
\bibitem{AA_2017}
H. Yan, S. P. Tan and Y. Z. He, A small-gain method for integrated guidance and control in terminal phase of reentry, Acta Astronautica, Vol. 132, 2017, pp. 282-292.
\bibitem{Shuttle}
J. C. Harpold and C. A. Graves, Shuttle entry guidance, Journal of the Astronautical Sciences, Vol. 27, No. 3, 1979, pp. 239-268.
\bibitem{Talole_2007}
S. E. Talole, J. Benito and K. D. Mease, Sliding mode observer for drag tracking in entry guidance, Proceedings of the AIAA Guidance, Navigation and Control Conference, Hilton Head, South Carolina, 2007, pp. 5122-5137.
\bibitem{Saraf_2004}
A. Saraf, A. Leavitt, D. T. Chen and K. D. Mease, Design and evaluation of an acceleration guidance algorithm for entry, Journal of Spacecraft and Rockets, 2004, Vol. 41, No. 6, 2004, pp. 986-996.
\bibitem{Mease_2002}
K. D. Mease, D. T. Chen, P. Teufel and H. Sch-ograve, Reduced-order entry trajectory planning for acceleration guidance, Journal of Guidance, Control and Dynamics, 2002, Vol. 25, No. 2, 2002, pp. 257-266.
\bibitem{Leavitt_2007}
J. A. Leavitt and K. D. Mease, Feasible trajectory generation for atmospheric entry guidance, Journal of Guidance, Control and Dynamics, Vol. 30, No. 2, 2007 pp. 472-481.
\bibitem{Wang_2013}
S. H. Wang, P. L. Fei, X. X. Liu and B. Zhang, A new evolved acceleration reentry guidance for reusable launch vehicles, Applied Mechanics and Materials, Vols. 380-384, 2013, pp. 576-580.
%\bibitem{liuhelong}
%H. L. Liu, Y. Z. He, H. Yan and S. P. Tan, Tether tension control law design during orbital transfer via small-gain theorem, Aerospace Science and Technology, No. 63, 2017, pp. 191-202.
%\bibitem{book_of_K.Khalil}
%H. K. Khalil, Nonlinear Systems, 3rd ed., Prentice-Hall, Upper Saddle River, NJ, 2002, Chap. 4.
%\bibitem{Small-GainTheorem}
%Z. P. Jiang, A. R. Teel, and L. Praly, Small-gain theorem for ISS systems and applications, Mathematics of Control, Signals, and Systems, Vol. 7, No. 2, 1994, pp. 95-120.
%\bibitem{Small-GainTheorem2}
%Z. P. Jiang, Iven M. Y. Mareels, A small-gain control method for nonlinear cascaded systems with dynamic uncertainties, IEEE Transactions on Automatic Control, Vol. 42, No. 3, 1997, pp. 292-308.
%\bibitem{Guo_2013}
%D. Y. Wang and M. W. Guo, Robust guidance law for drag tracking in mars atmospheric entry flight, Proceedings of the 33rd Chinese Control Conference, Nanjing, China, 2014, pp. 697-702.
\bibitem{Lu_1997}
P. Lu, Entry guidance and trajectory control for reusable launch vehicle, Journal of Guidance, Control and Dynamics, Vol. 20, No. 1, 1997, pp. 143-149.
\bibitem{Benito_2008}
J. Benito and K. D. Mease, Nonlinear predictive controller for drag tracking in entry guidance, Proceedings of the 2008 AIAA/AAS Astrodynamics Specialist Conference and Exhibit, Honolulu, Hawaii, AIAA-2008-7350.
\bibitem{Guo}
M. W. Guo and D. Y. Wang, Guidance law for low-lifting skip reentry subject to control saturation based on nonlinear predictive control, Aerospace Science and Technology, Vol. 37, 2014, pp. 48-54.
\bibitem{XiaYuanQing}
Y. Q. Xia, R. F. Chen, F. Pu and L. Dai, Active disturbance rejection control for drag tracking in mars entry guidance, Advances in Space Research, Vol. 53, 2014, pp. 853-861.
\bibitem{Xia_2013}
R. F. Chen and Y. Q. Xia, Drag-based entry guidance for mars pinpoint landing, Proceedings of the 32nd Chinese Control Conference, Xi'an, China, 2013, pp. 5473-5478.
\bibitem{Yan}
H. Yan and Y. Z. He, Drag-tracking guidance for entry vehicles without drag rate measurement, Aerospace Science and Technology, Vol. 43, 2015, pp. 372-380.
\bibitem{Wen_2011}
C. Y. Wen, J. Zhou, Z. T. Liu, and H. Y. Su, Robust adaptive control of uncertain nonlinear systems in the presence of input saturation and external disturbance, IEEE Transactions on Automatic Control, Vol. 56, No. 7, 2011, pp. 1672-1678.
\bibitem{Liang_2015}
X. L. Liang, M. Z. Hou, and G. R. Duan, Adaptive dynamic surface control for integrated missile guidance and autopilot in the presence of input saturation, Journal of Aerospace Engineering, Vol. 28, No. 5, 2015, pp. 04014121.
\bibitem{Khalil_1992}
F. Esfandiari and H. K. Khalil, Output feedback stabilization of fully linearizable systems, International Journal of Control, Vol. 56, No. 5, 1992, pp. 1007-1037.
\bibitem{Khalil_1999}
A. N. Atassi and H. K. Khalil, A separation principle for the stabilization of a class of nonlinear systems, IEEE Transactions on Automatic Control, Vol. 44, No. 9, 1999, pp. 1672-1687.
\bibitem{Lu}
P. Lu, Predictor-corrector entry guidance for low-lifting vehicles, Journal of Guidance, Control and Dynamics, Vol.31, No. 4, 2008, pp. 1067-1075.
\bibitem{Guo_2013}
D. Y. Wang and M. W. Guo, Robust guidance law for drag tracking in mars atmospheric entry flight, Proceedings of the 33rd Chinese Control Conference, Nanjing, China, 2014, pp. 697-702.
\end{thebibliography}
\end{document}